\documentclass[twocolumn]{IEEEtran} 

\usepackage{amsmath,epsfig}
\usepackage[T1]{fontenc}
\usepackage{psfrag}
\usepackage{array}
\usepackage{cite}
\usepackage{amsfonts}
\usepackage{algorithm}
\usepackage{algorithmic}
\usepackage[utf8]{inputenc} 

\usepackage{amsthm}

\newtheorem{prop}{Proposition}

\usepackage{amsmath}
\usepackage{color}
\usepackage{fancybox}
\usepackage{array}
\usepackage[dvipsnames]{xcolor}
\definecolor{MyGreen}{rgb}{0.4660 0.6740 0.1880}
\definecolor{MyPink}{rgb}{1 0.6 0.7373}

\usepackage{graphicx}

\usepackage[caption=false,listofformat=subsimple,labelformat=simple]{subfig}

\newcolumntype{P}[1]{>{\centering\arraybackslash}p{#1}}
\newcolumntype{M}[1]{>{\centering\arraybackslash}m{#1}}

\usepackage{circuitikz}
\usepackage{tikz}
\usetikzlibrary{fit,positioning}
\usetikzlibrary{shapes.geometric}
\usetikzlibrary{decorations.pathreplacing,calligraphy}
\usepackage{adjustbox}

\usepackage[outline]{contour}
\contourlength{0.75pt}

\usepackage{float}
\usepackage{url}

\newcommand{\B}[1]{{\mathbf{#1}}}

\DeclareFontFamily{OT1}{pzc}{}
\DeclareFontShape{OT1}{pzc}{m}{it}%
             {<-> s * [1.28] pzcmi7t}{} 
\DeclareMathAlphabet{\freqDom}{OT1}{pzc}
                                 {m}{it}
\DeclareMathOperator{\Tr}{Tr}
\DeclareMathOperator{\diag}{diag}
\usepackage{amsbsy} 

\usepackage{tcolorbox}
\tcbuselibrary{listingsutf8} 

\usepackage{acronym}
\usepackage{soul}
\usepackage{amssymb}

\makeatletter
\def\blfootnote{\xdef\@thefnmark{}\@footnotetext}
\makeatother

\makeatletter
\def\blfootnote{\xdef\@thefnmark{}\@footnotetext}
\makeatother

\acrodef{3GPP}[3GPP]{3rd Generation Partnership Project}
\acrodef{5G}[5G]{fifth-generation}
\acrodef{6G}[6G]{sixth-generation}
\acrodef{ADC}[ADC]{Analog-to-Digital Converter}
\acrodef{AO}[AO]{Alternating Optimization}
\acrodef{AoA}[AoA]{Angle of Arrival}
\acrodef{AoD}[AoD]{Angle of Departure}
\acrodef{APN}[APN]{Analog Precoding Network}
\acrodef{ARV}[ARV]{Array Response Vector}
\acrodef{ASK}[ASK]{Amplitude-Shift Keying}
\acrodef{AWGN}[AWGN]{Additive White Gaussian Noise}
\acrodef{BER}[BER]{Bit Error Ratio}
\acrodef{BF}[BF]{Beamforming}
\acrodef{BFN}[BFN]{Beamforming Network}
\acrodef{BS}[BS]{Base Station}
\acrodef{CB}[CB]{conjugate beamforming}
\acrodef{COMP}[COMP]{Covariance OMP}
\newacro{CSI}{Channel State Information}
\acrodef{DAC}[DAC]{Digital to Analog Converter}
\acrodef{DAS}[DAS]{distributed antenna system}
\acrodef{DBS}[DBS]{Distance-Based Scheduling}
\acrodef{DCS}[DCS]{Digital Communication System}
\acrodef{DCOMP}[DCOMP]{Dynamic COMP}
\newacro{DFT}{Discrete Fourier Transform}
\acrodef{DL}[DL]{downlink}
\acrodef{DOA}[DOA]{Direction Of Arrival}
\acrodef{DoF}[DoF]{degrees-of-freedom}
\acrodef{DPC}[DPC]{dirty-paper coding}
\acrodef{ELAA}[ELAA]{extremely-large aperture array}
\acrodef{ESD}[ESD]{Energy Spectral Density}
\acrodef{FCC}[FCC]{Federal Communications Commission}
\newacro{FDD}{Frequency-Division Duplex}
\acrodef{FF}{far-field}
\acrodef{FR}{Frequency Range}
\acrodef{FLS}[FLS]{Front Line Scheduling}
\acrodef{FSK}[FSK]{Frequency-Shift Keying}
\acrodef{FT}[FT]{Fourier Transform}
\acrodef{HT}[HT]{Hilbert Transform}
\acrodef{HI}[HI]{Harmonic Interference}
\acrodef{ICI}[ICI]{Inter-Carrier Interference}
\acrodef{gMIMO}[gMIMO]{gigantic MIMO}
\acrodef{IL}[IL]{Insertion Losses}
\acrodef{ISI}[ISI]{Inter-Symbol Interference}
\acrodef{ITU}[ITU]{International Telecommunications Union}
\acrodef{IUI}[IUI]{inter-user interference}
\acrodef{JSDM}[JSDM]{Joint Spatial Division and Multiplexing}
\acrodef{LBFN}[LBFN]{Linear Beamforming Network}
\acrodef{LLF}[LLF]{Log-Likelihood function}
\acrodef{LMD}[LMD]{Linearly Modulated Digital}
\acrodef{LoS}[LoS]{Line-of-Sight}
\newacro{MA}{Modular Array}
\newacro{MAP}{Maximum A Posteriori}
\acrodef{MIMO}[MIMO]{Multiple-Input Multiple-Output}
\newacro{ML}{Maximum Likelihood}
\newacro{MMSE}{Minimum Mean Squared Error}
\acrodef{MMV}[MMV]{Multiple Measurement Vector}
\acrodef{mmWave}[mmWave]{millimeter wave}
\acrodef{MRC}[MRC]{Maximum Ratio Combining}
\acrodef{MRT}[MRT]{Maximum Ratio Transmission}
\newacro{MSE}{Mean Squared Error}
\acrodef{MUSIC}[MUSIC]{MUltiple SIgnal Classification}
\acrodef{NF}{near-field}
\acrodef{NLoS}[NLoS]{Non Line-of-sight}
\acrodef{NMSE}{Normalized Mean Squared Error}
\acrodef{OFDM}[OFDM]{Orthogonal Frequency-Division Multiplexing}
\acrodef{OFDMA}[OFDMA]{Orthogonal Frequency-Division Multiple Access}
\acrodef{OMP}[OMP]{Orthogonal Matching Pursuit}
\acrodef{OSMP}[OSMP]{Orthogonal Subspace Matching Pursuit}
\acrodef{PA}[PA]{Power Amplifier}
\acrodef{PS}[PS]{Phase Shifter}
\acrodef{PSK}[PSK]{Phase-Shift Keying}
\acrodef{PW}{planar wavefront}
\acrodef{QAM}[QAM]{Quadrature Amplitude Modulation}
\acrodef{RF}[RF]{Radio Frequency}
\acrodef{RFC}[RFC]{Rayleigh Fading Channel}
\acrodef{RSS}[RSS]{Rectangular-Search Scheduler}
\acrodef{SDMA}[SDMA]{space-division multiple access}
\acrodef{SE}[SE]{spectral efficiency}
\acrodef{SER}[SER]{Symbol Error Rate}
\acrodef{SIC}[SIC]{Successive Interference Cancellation}
\acrodef{SR}[SR]{Sideband Radiation}
\acrodef{SINR}[SINR]{Signal-to-Interference-plus-Noise Ratio}
\acrodef{SLL}[SLL]{Side-Lobe Level}
\acrodef{SOCP}[SOCP]{Second-Order Cone Program}
\acrodef{SOMP}[SOMP]{Simultaneous-Orthogonal Matching Pursuit}
\acrodef{SPDT}[SPDT]{Single-Pole-Double-Throw}
\acrodef{SPST}[SPST]{Single-Pole-Single-Throw}
\acrodef{SR}[SR]{Sideband Radiation}
\acrodef{SS}[SS]{Spatial Smoothing}
\acrodef{SNR}[SNR]{Signal-to-Noise Ratio}
\acrodef{SUS}{successive user selection}
\acrodef{SW}{spherical wavefront}
\newacro{TDD}{Time-Division Duplex}
\acrodef{TM}[TM]{Time Modulation}
\acrodef{TMA}[TMA]{Time-Modulated Array}
\acrodef{ULA}[ULA]{Uniform Linear Array}
\acrodef{UM-MIMO}[UM-MIMO]{Ultra Massive MIMO}
\acrodef{UPA}[UPA]{Uniform Planar Array}
\acrodef{UPW}[UPW]{Uniform Planar Wave}
\acrodef{USW}[USW]{Uniform Spherical Wave}
\acrodef{VGA}[VGA]{Variable Gain Amplifier}
\acrodef{VPS}[VPS]{Variable Phase Shifter}
\acrodef{VR}[VR]{visibility regions}
\acrodef{XL}[XL]{extra-large}
\acrodef{XL-array}[XL-array]{extra-large array}
\acrodef{XL-MIMO}[XL-MIMO]{extra-large Multiple-Input Multiple-Output}
\acrodef{XL-ULA}[XL-ULA]{extra-large uniform linear array}
\acrodef{ZF}[ZF]{Zero-Forcing}

\begin{document}

\title{{User Selection in Near-Field Gigantic MIMO Systems with Modular Arrays}}

\author{José P. Gonz\'alez-Coma, Santiago Fern\'andez, and F. Javier L\'opez-Mart\'inez}

\maketitle
\begin{abstract}
\acp{MA} are a promising architecture to enable multi-user communications in next-generation multiple-input multiple-output (MIMO) systems based on extra-large (XL) or gigantic MIMO (gMIMO) deployments, trading off improved spatial resolution with characteristic interference patterns associated with grating lobes. In this work, we analyze whether \acp{MA} can outperform conventional collocated deployments, in terms of achievable sum-\ac{SE} and served users in a multi-user downlink set-up. First, we provide a rigorous analytical characterization of the inter-user interference for modular gMIMO systems operating in the near field. Then, we leverage these results to optimize the user selection and precoding mechanisms, designing two algorithms that largely outperform existing alternatives in the literature, with different algorithmic complexities. Results show that the proposed algorithms yield over $70\%$ improvements in achievable sum-spectral efficiencies compared to the state of the art. We also illustrate how \acp{MA} allows us to serve a larger number of users thanks to their improved spatial resolution, compared to the collocated counterpart.
 \end{abstract}
\blfootnote{\noindent Manuscript received January 09, 2025; revised March 13, 2025. The review of this paper was coordinated by XXXX. This work is supported by grant PID2023-149975OB-I00 (COSTUME) funded by MICIU/AEI/10.13039/501100011033, and by ERDF/EU. The authors thank the Defense University Center at the Spanish Naval Academy for their support. This work has been submitted to the IEEE for publication. Copyright may be transferred without notice, after which this
version may no longer be accessible.}

\blfootnote{\noindent J.P. Gonz\'alez-Coma is with the Defense University Center at the Spanish Naval Academy, 36920 Marín, Spain. Contact email: jose.gcoma@cud.uvigo.es.}

\blfootnote{\noindent S. Fern\'andez and F.J. L{\'o}pez-Mart{\'i}nez are with the Dept. Signal Theory, Networking and Communications, Research Centre for Information and Communication Technologies (CITIC-UGR), University of Granada, 18071, Granada, Spain, and also with the Communications and Signal Processing Lab, Telecommunication Research Institute (TELMA), Universidad de M\'alaga, M\'alaga, 29010, (Spain). }
\acresetall
\section{Introduction}

\IEEEPARstart{M}{assive} \ac{MIMO} technology is one of the key pillars for the success of \ac{5G} wireless networks, and continuous efforts are made to improve its operation in the subsequent releases by \ac{3GPP} towards \ac{5G}-advanced \cite{Liberg2024}. As the different pathways to conceive the core technologies within the future \ac{6G} standard are being discussed, the integration of an even larger number of antennas is facilitated \textcolor{black}{as} operational frequency grows. Now, while academic research boldly advocates for moving beyond \ac{mmWave} into subTHz bands \cite{Wang2023,Ning2023} to enable the \ac{UM-MIMO} concept, this could exacerbate the practical limitations of \ac{5G} commercial deployments in the \ac{FR} of \ac{mmWave} frequencies (FR2) \cite{Rochman2021}. Conversely, industry plays it safe and stands up for modest increases in the operational frequencies; in this sense, there seems to be a consensus that the so-called upper midband (FR3) covering the 7-24 GHz range will play a central role in early \ac{6G} \cite{Cui2023,Emil2024}, as confirmed by regulatory bodies such as the \ac{FCC} and \ac{ITU} \cite{ITU2023,Kang2024}.

The evolution of antenna array technologies in the upper-midband allows to integrate thousands of antennas in a reasonable area at the \ac{BS} side. This, combined with the use of distributed array deployments different from conventional collocated architectures \cite{Jeon2021,Li2022,Li2024}, opens the door to extend the \ac{NF} features that enable beamfocusing to the entire operational range of the \ac{BS} in \ac{6G} \cite{Emil2024}. In this sense, the evolution of conventional collocated arrays to sparse or distributed deployments is shown to be beneficial for use cases like enhanced mobile broadband, localization and sensing. Recently, \acp{MA} have been proposed \cite{Jeon2021,Li2022} in the context of \ac{XL}-\ac{MIMO} or \ac{gMIMO} evolutions. This newly proposed architecture, \textcolor{black}{which includes the uniform sparse and collocated arrays as special cases,} deploys multiple modules or sub-arrays with a moderate number of elements, classically \textcolor{black}{(but not necessarily)} spaced $\lambda/2$, with $\lambda$ being the operational signal wavelength. Now, individual modules are separated much more than $\lambda$, thus increasing the overall array aperture. \textcolor{black}{In some instances, conventional XL collocated arrays or sparse arrays cannot be deployed due to physical constraints, e.g. because of the irregular shape of mounting structures like facades with windows, and MAs become a feasible solution. In other cases, MAs are shown to} offer an alternative to implement coordinated transmission between sub-arrays with minimal synchronization requisites and a reduced backhaul overhead, \textcolor{black}{which can be beneficial to shape interference and to improve localization and sensing performances \cite{Emil2024,Li2024} compared to collocated implementations}. Noteworthy, \acp{MA} allow the operation in the \ac{NF} for a larger distance, even when individual sub-arrays may operate in the \ac{FF} zone. This offers an improved spatial resolution, at the expense of the appearance of grating lobes due to the large inter-module separation \cite{Li2023}, which cause characteristic interference patterns that affect multi-user operation. 

To avoid achievable sum-\ac{SE} performance plummeting as the \ac{BS} intends to serve a larger number of users, user selection\footnote
{Also referred to as user scheduling, user grouping, or user dropping in the literature.} plays a pivotal role in the operation of multi-user communication systems. Selecting the set of users that will be allocated for transmission in the same time and frequency resource, and designing the corresponding precoding vectors and power allocation policies, is known to be a highly complex combinatorial and non-convex optimization problem, even when linear precoders are used \cite{GuUtDi09}. The special features of \ac{NF} propagation allow to better exploit the available spatial \ac{DoF} \cite{Kosasih2024} which, in turn, allows for improved performance when taken into consideration in the user selection process \cite{GonzalezComa2021,Souza2023,Anaya2023}. For the specific cases of \acp{MA} \textcolor{black}{and sparse arrays}, this problem remains largely unexplored, and was only addressed in \cite{Li2023}, where a greedy approach was explored to relax the computational burden. While showing some sum-\ac{SE} improvements over the random user grouping case, it fails to provide a convincing performance when compared to state-of-the-art alternatives.

In this paper, we address the problem of user selection in the context of \ac{gMIMO} systems equipped with \acp{MA}. First, we perform an analytical characterization of the multi-user interference patterns with \ac{MA} deployments. This allows to obtain important insights related to the behavior of interference, which are later leveraged to design algorithms for joint user selection and precoding in this context. 
 The performance and complexity of the proposed algorithms are evaluated in practical scenarios of gMIMO deployments in the upper midband, showing remarkable performance improvements over the greedy approach in \cite{Li2023}.
 
{The remainder of this paper is structured as follows:  In Section \ref{Sec:SystemModel}, the system model for a multi-user gMIMO system under \ac{SW} propagation is introduced. Then, the analytical characterization of the interference patterns is carried out in Section \ref{sec:interferenceModular}. The algorithms for user selection and precoding design are presented in Section \ref{Sec:US}, and their performance and complexity are assessed in Section \ref{Sec:NR}. Finally, conclusions are outlined in Section \ref{Sec:Conclusions}.}

\textit{Notation:} $a$ is a scalar, $\mathbf{a}$ is a vector, and $\mathbf{A}$ is a matrix. 
Transpose and conjugate transpose of $\mathbf{A}$ are denoted by $\mathbf{A}^{\operatorname{T}}$ and $\mathbf{A}^{\operatorname{H}}$, respectively.  Calligraphic
letters, e.g., $\mathcal{A}$ denote sets and sequences. $\mid \mathcal{A} \mid$ represents the set cardinality and $\mathcal{A}\setminus \lbrace b \rbrace$ stands for the exclusion of $b$ from $\mathcal{A}$, $\odot$ and $\otimes$ represent the Hadamard and the Kronecker product, respectively.
Finally, the expectation operator is $E[\cdot]$ and $\parallel \cdot  \parallel$ denotes the Euclidean norm.

\section{System Model}
\label{Sec:SystemModel}
Let us consider a modular \ac{gMIMO} setup deployed at the \ac{BS}, where $N$ modules composed by $M$ antenna elements each form the array, hence consisting of a total of $N\times M$ elements as illustrated in Fig. \ref{fig:XL_MIMO_SystemModel_2users_new}. Without loss of generality, a modular \ac{ULA} is assumed to be placed along the $y$-axis, symmetric around the origin. 
For the sake of notation simplicity, we consider that the module index $n$ and the antenna index $m$ for each module belong to the integer sets $\mathcal{N} = \{ 0, \pm 1, \cdots , \pm (N-1)/2 \}$ and $\mathcal{M} = \{ 0, \pm 1, \cdots , \pm (M-1)/2 \}$. Consequently, the position of the $m$-th element within module $n$ is
$[0, (nS + m)d]^T$, $\forall n \in \mathcal{N}$, and $\forall  m \in \mathcal{M}$. In the last expression, it is assumed a typical inter-element spacing for antennas within each module of half of the signal wavelength $\lambda$ denoted by $d$, and $S$ is the modular separation parameter, that may depend on the discontinuous surface of the practical installation structure, with $S \geq M$.
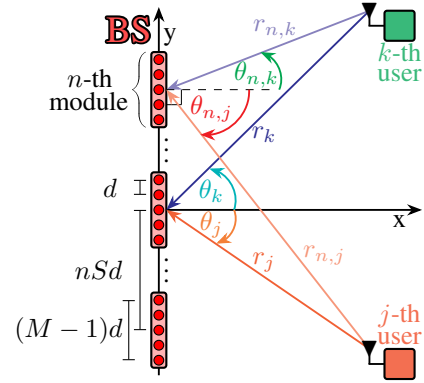
\begin{figure}[t]
\centering 
\captionsetup{justification=centering}
\begin{tikzpicture}
			\let\radius\undefined
		\newlength{\radius}
		\setlength{\radius}{0.65mm}
		
		\def\NN{0, 0.4, 0.8}
		\def\MM{0, 0.2, 0.4, 0.6, 0.8} 
		\def\S{4}
		
		\draw [line width=0.25mm,-{Stealth[length=2mm]}] (0,-0.2) -- (0,4.7);
		\node[centered, centered, black] at (0.15,4.3) {y};
		
		\draw [line width=0.25mm,-{Stealth[length=2mm]}] (0,2) -- (3.5,2);		
		\node[centered, centered, black] at (3.2,1.85) {x};
		
		\foreach \mm in {0, 1.6, 3.2}
		\draw[black,fill=red!30,rounded corners=1,thick]
		(-0.1,-0.1+\mm) rectangle ++(0.2,1);
		
		\foreach \n in \NN 
		\foreach \m in \MM 			
		\draw [black, fill=red] (0,\n*\S+ \m) circle (\radius);   
		\node[] at (0.1,1.3) {\vdots};
		\node[] at (0.1,1.2+1.7) {\vdots};
		
		\draw [|-|](-0.4,0) -- (-0.4,0.8);
		\node[centered, centered, black] at (-1.2,0.4) {$(M-1)d$};
		
		\draw [|-|](-0.25,2.2) -- (-0.25,2.4);
		\node[centered, centered, black] at (-0.65,2.3) {$d$};
		
		\draw [|-|](-0.25,2) -- (-0.25,0.4);
		\node[centered, centered, black] at (-0.8,1.2) {$n S d$};
		
		\node[centered, centered, red!70] at (-0.4,4.4) {\contour{black}{\Large{BS}}};
		
		\draw [pen colour={black},
		decorate, decoration = {calligraphic brace,	raise=5pt,
			amplitude=5pt}] (0,3.1) --  (0,4.1)
		node[pos=0.65,left=14pt,black]{$n$-th}
		node[pos=0.35,left=8pt,black]{module};		
		
		\node[centered, centered, Blue] at (2.1,3) {$r_{k,0,0}\triangleq r_k$};
		\draw [Blue, line width=0.25mm,-{Stealth[length=2mm]}] (2.8,4.65) -- (0.1,2);
		
		\node[centered, centered, BlueGreen] at (1.85,2.25) {$\theta_{k,0,0}\triangleq \theta_k$};
		\draw [BlueGreen, line width=0.25mm,-{Stealth[length=2mm]}] (1,2) to [bend right=50] (0.67,2.55);		
		
		\node[centered, centered, Blue!50] at (1.55,4.4) {$r_{k,n,m}$};
		\node[centered, centered, Green] at (1.35,3.8) {$\theta_{k,n,m}$};
		\draw [Blue!50, line width=0.25mm,-{Stealth[length=2mm]}] (2.8,4.65) -- (0.1,3.6);
		\draw [Green, line width=0.25mm,-{Stealth[length=2mm]}] (1.6,3.6) to [bend right=50] (1.35,4.1);		
		
		\draw [line width=0.1mm, dashed](0.1,3.6) -- (1.5,3.6);
		\draw [line width=0.05mm](0.3,3.6) -- (0.3,3.4);
		\draw [line width=0.05mm](0.3,3.4) -- (0.1,3.4);
		
		\draw[black,fill=Green!70,rounded corners=1,thick]
		(3,4.25) rectangle ++(0.4,0.4);
		\node[Green!70, centered, centered] at (3.2,4.1) {$k$-th};
		\node[Green!70, centered, centered] at (3.2,3.8) {user};
		\node[centered, centered] at (2.8,4.65) {$\blacktriangledown$};
		
		\draw [line width=0.25mm, black ] (2.7875,4.45) -- (3,4.45);
		\draw [line width=0.25mm, black ] (2.8,4.65) -- (2.8,4.45);

		\draw [RedOrange, line width=0.25mm,-{Stealth[length=2mm]}] (2.8,0.15) -- (0.1,2);
		\node[centered, centered, RedOrange] at (1.4,1.3) {$r_j$};
		
		\draw [Orange, line width=0.25mm,-{Stealth[length=2mm]}] (1,2) to [bend left=50] (0.76,1.55);		
		\node[centered, centered, Orange] at (0.7,1.8) {$\theta_j$};
	
	\draw [Red!50, line width=0.25mm,-{Stealth[length=2mm]}] (2.8,0.15) -- (0.1,3.6);
	\draw [Red, line width=0.25mm,-{Stealth[length=2mm]}] (1.2,3.6) to [bend left=50] (0.55,3);		
	\node[centered, centered, Red!50] at (2.2,1.4) {$r_{j,n,m}$};
	\node[centered, centered, Red] at (0.7,3.35) {$\theta_{j,n,m}$};
	
	\draw[black,fill=Red!70,rounded corners=1,thick]
	(3,-0.25) rectangle ++(0.4,0.4);
	\node[Red!70, centered, centered] at (3.2,0.6) {$j$-th};
	\node[Red!70, centered, centered] at (3.2,0.3) {user};
	\node[centered, centered] at (2.8,0.15) {$\blacktriangledown$};
	
	\draw [line width=0.25mm, black ] (2.7875,-0.05) -- (3,-0.05);
	\draw [line width=0.25mm, black ] (2.8,0.15) -- (2.8,-0.05);

	\end{tikzpicture}
\caption{\small A modular \ac{gMIMO} system with $N$ modules and $M$ antennas in each module communicating with single-antenna users.}
\label{fig:XL_MIMO_SystemModel_2users_new}
\end{figure}

The position of the $k$-th user is then determined by its angle with respect to the positive $x$-axis $\theta_k \in [ -\pi/2, \pi/2]$ and the distance to the $m$-th element in module $n$, such that
\begin{equation}
    r_{k,n,m} = \sqrt{r_k^2 - 2r_k(nS + m)d\sin{\theta_k} + \left((nS + m)d\right)^2},
    \label{eq:r_nm}
\end{equation}
where $r_k$ denotes its distance from the center of the antenna array. \textcolor{black}{In the sequel, for the sake of notational simplicity, we omit the dependency on $k$ in $r_{k,m,n}$ and $\theta_{k,m,n}$ unless specifically required to avoid ambiguity.}

Based on this formulation and assuming that the \ac{LoS} component dominates the channel vector response \cite{Zhou2015}, it is possible to write the channel model to accurately characterize the \textcolor{black}{complex-valued signal across the array elements} for every $k$-th user as
\begin{equation}
    \mathbf{h}(r,\theta) = \frac{\sqrt{\beta_0}}{r} \mathbf{a}(r,\theta),
    \label{eq:channel}
\end{equation}
%
%
\noindent where $\beta_0$ denotes the reference channel power gain at the distance of 1 meter (m), and the \ac{ARV} \textcolor{black}{
$\mathbf{a} \in \mathbb{C}^{N M\times 1}$ for every $k$-th user} can be written as
%
%
\textcolor{black}{
\begin{equation}
    \mathbf{a}(r,\theta) = \left[\frac{r}{ r_{n,m}}e^{-j \tfrac{2 \pi}{\lambda} r_{n,m}} \right]_{n \in \mathcal{N}, m \in \mathcal{M}}. 
\end{equation}
}

We also consider that the \textcolor{black}{amplitude variations over all array elements} can be neglected if $r \geq 1.2\Delta_{NM}$ \cite{Bjornson2021}, where $\Delta_{NM}= [(N-1)S + (M-1)]d$ is the total physical size of the modular \ac{gMIMO} antenna array. Then, the position of the $k$-th user in \eqref{eq:r_nm} can be simplified to \textcolor{black}{ $\mathbf{a}(r,\theta) \approx [e^{-j \frac{2 \pi}{\lambda} r_{n,m}} ]_{n \in \mathcal{N}, m \in \mathcal{M}}$}. Based on this observation and recalling that the boundary between \ac{FF} and \ac{NF} regions of antenna arrays can be pointed out based on the classical Rayleigh distance, i.e., $r_{\rm Ray} = 2\Delta_{NM}^2/\lambda$, it is possible to express the \ac{ARV} $\mathbf{a}(r,\theta)$ as a function of $r$ based on the \ac{PW} and the \ac{SW} models for modular \ac{gMIMO} antenna arrays.

\subsection{Array response vector for modular arrays}

In this work, we are interested in arrays in which the module sizes are small compared to the dimensions of the entire array, i.e., \textcolor{black}{ $2\Delta_{M}^2 \leq \lambda \ r < 2\Delta_{NM}^2$}, where $\Delta_{M} = (M-1)d$ is the physical size of each module. In such a case, for a practical range of usage, the users are located within the \ac{NF} region of the whole array, but within the \ac{FF} region of each array module. In other words, there is no significant distance difference between users and adjacent antennas within a particular module. In contrast, the distance and \ac{AoA} vary across different modules. In this way, the distance from the $n$-th module to a user located at a distance $r$ and at an angle $\theta$ reads as
\begin{equation}
    r_{n} = \sqrt{r^2 - 2 r n S d \sin{\theta} + (n S d)^2} \quad \forall n \in \mathcal{N}\textcolor{black}{.}
    \label{eq:distanceNmodule}
\end{equation}
Moreover, since $\theta_n$ represents the angle of the user to the center of each $n$-th module, it is defined \textcolor{black}{by}
\begin{equation}
    \sin{\theta_{n}} = \frac{r \sin{\theta} - n S d}{r_{n}} \quad \forall n \in \mathcal{N}.
    \label{eqtheta}
\end{equation}
Based on this geometrical model, the \ac{ARV} can be compactly expressed as
\begin{align}
    \B{a}(r,\theta)=\left(\B{q}(r,\theta)\otimes\B{1}_M\right)\odot \textcolor{black}{\B{u}(\theta)},
    \label{eq:ARV_General}
\end{align}
where $\B{q}(r,\theta)\in \mathbb{C}^{N \times 1}$ comprises the near-field delays as
\begin{align*}
    \B{q}(r,\theta) =  \left[e^{-j \frac{2 \pi}{\lambda} r_{1}},\ldots,e^{-j \frac{2 \pi}{\lambda} r_{N}}\right]^T,
\end{align*}
$\B{1}_M$ is a column vector containing $M$ ones, and \textcolor{black}{$\mathbf{u}(\theta) = [ \B{b}^T(\theta_1),\ldots,\B{b}^T(\theta_N)]^T \in \mathbb{C}^{NM \times 1} $} with $\mathbf{b}(\theta_n) = [ e^{j \tfrac{2\pi}{\lambda} m d \sin (\theta_n)}]_{m \in \mathcal{M}} \in \mathbb{C}^{M \times 1}$ for the $n$-th module represent the far-field effects within each of the modules\footnote{\textcolor{black}{The notation $\mathbf{u}(\theta)$ is used for the sake of simplicity in the sequel; however, note that $\theta_n$ depends on $r$ through \eqref{eqtheta}.}}. In the scenarios where the users are located on the \ac{FF} region of the whole array, i.e. $r \geq r_{\rm Ray}$, the first-order Taylor expansion can be used to approximate the distance expression in \eqref{eq:distanceNmodule} as a linear function of the antenna index, so that $r_{n} \approx r -n S d\sin (\theta)$ and we get
\begin{align}
    \B{q}(r,\theta) =   e^{-j \frac{2 \pi}{\lambda} r} \left[e^{-j \frac{2 \pi}{\lambda}  S d \sin (\theta)\frac{N-1}{2}},\ldots,e^{j \frac{2 \pi}{\lambda} S d \sin (\theta)\frac{N-1}{2}}\right]^T.    
    \label{eq:PW_delays}
\end{align}

Moreover, in this case, the angle is common for all the modules, i.e. $\theta_n\approx\theta$, such that $\textcolor{black}{\B{u}(\theta)}=\B{1}_N\otimes \B{b}(\theta)$, and the array response vector simplifies to $\B{a}(r,\theta)=\B{q}(r,\theta)\otimes\B{b}(\theta)$. 

\subsection{Problem formulation}
%

In this work, \textcolor{black}{we focus on the downlink of a modular gMIMO system, where the \ac{BS} communicates with $K$ single-antenna users. We denote the data symbol intended for user $k$ as $s_k$,  with $k\in\mathcal{K} = \{1, \ldots, K \}$, and consider zero-mean data symbols $s_k\in \mathcal{N}_{\mathbb{C}}(0,1)$}. 
To enable the system's multi-user capabilities, the symbols are linearly processed with the unit-norm precoding vectors $\mathbf{p}_k\in \mathbb{C}^{NM\times 1}$. Thus, the transmitted signal is $\mathbf{x}=\sum_{k\in\mathcal{K}}\sqrt{p_{k}}\mathbf{p}_{k}s_k$, where $p_k$ is the power allocated to the user, and we consider the total transmit power constraint $\sum_{k\in\mathcal{K}}p_{k}\leq P_{\text{TX}}$. The transmitted signal then inputs the wireless channel represented by the channel vector $\mathbf{h}(r_k,\theta_k)$ from \eqref{eq:channel} leading to the received signal for the $k$-th user
\begin{align}
    y_k & = \mathbf{h}^H(r_k,\theta_k)\mathbf{x}+z_k=\mathbf{h}^H(r_k,\theta_k)\sum\nolimits_{k\in\mathcal{K}}\sqrt{p_{k}}\mathbf{p}_{k}s_k+z_k\notag\\
    & = \mathbf{h}^H(r_k,\theta_k)\sqrt{p_{k}}\mathbf{p}_{k}s_k+\mathbf{h}^H(r_k,\theta_k)\sum\nolimits_{i\neq k}\sqrt{p_{i}}\mathbf{p}_{i}s_i+z_k\label{eq:sigr}
\end{align}
where $z_k\in \mathcal{N}_{\mathbb{C}}(0,\sigma^2)$ is the noise. Note that the first term in \eqref{eq:sigr} is the desired signal whereas the second one collects the \ac{IUI}. When the number of users grows large compared to $NM$, the achievable \ac{SE} is limited by the \ac{IUI}, as can be noticed from
\begin{align}\label{eq:rate}
R_k=
     \text{log}_2 \left(1+ \frac{p_k|\mathbf{h}^H(r_k,\theta_k)\mathbf{p}_k|^2}{\sum\nolimits_{i\in \mathcal{S},i\neq k}p_{i}|\mathbf{h}^H(r_k,\theta_k)\mathbf{p}_{i}|^2+\sigma^2}
    \right),
\end{align}
where $\mathcal{S} \subseteq \mathcal{K}$  is the set of served users, i.e., those with ${p}_{i}>0$. Two critical questions should then be addressed to solve this problem. First, it is of key importance to accurately characterize the \ac{IUI} according to the spatial positions of the users and then, according to such information, determine the set of users that are promising candidates to be served by the BS in each particular time-frequency resource block. The optimization problem is then formulated to maximize the achievable sum-\ac{SE} as
\begin{equation}
    \max_{\{\mathbf{p}_k\}_k\in\mathcal{S}}\sum_{k\in \mathcal{S}}R_k\quad\text{s.t.}\quad \sum_{k\in \mathcal{S}}p_k\leq P_{\text{TX}}.
    \label{eq:problem}
\end{equation}
\textcolor{black}{Observe that we aim at designing a set $\mathcal{S}$, and the corresponding precoding vectors, for a given time-frequency resource block and user set $\mathcal{K}$. This is compatible with a broad variety of higher-level resource allocation policies that define the resource blocks and candidate user set $\mathcal{K}$.}

In the next section, we analyze the dependence of \ac{IUI} on the user positions to provide a first approach to solving the combinatorial problem in \eqref{eq:problem}.

\section{Interference characterization in MAs}
\label{sec:interferenceModular}

In this section, we study the consequences of employing a modular array architecture in terms of inter-user interference. \textcolor{black}{Let us begin by considering a baseline scenario where the BS is only allowed to use the \ac{MRT} precoder scheme to communicate, where all users are served and allocated the same power budget. Then, \eqref{eq:rate} reads as
\begin{align}
    R_k&=\text{log}_2 \left(1+ \frac{\frac{1}{MN}|\mathbf{h}^H(r_k,\theta_k)\mathbf{h}(r_k,\theta_k)|^2}{\frac{1}{MN}\sum_{i\in \mathcal{S},i\neq k}|\mathbf{h}^H(r_k,\theta_k)\mathbf{h}(r_i,\theta_i)|^2+\sigma^2}
    \right)\notag\\
     &=\text{log}_2 \left(1+ \frac{I^2(k,k)}{\sum\nolimits_{i\in \mathcal{S},i\neq k}I^2(i,k)+\sigma^2\frac{r^4}{\beta_0^2}}
    \right),
    \label{eq:rateMRT}
\end{align}
where $I(j,k)=\frac{1}{MN}|\B{a}^H(r_j,\theta_j)\B{a}(r_k,\theta_k)|$ coincides with the \ac{IUI} up to some scaling factor. Consequently, we evaluate the interference $I(j,k)$} caused by a user $k$ located at a distance $r_k$ and angle $\theta_k$, to a different user $j$ in the position given by $r_j$ and $\theta_j$ from the general framework granted by the \ac{ARV} in \eqref{eq:ARV_General}, as shown in Fig. \ref{fig:XL_MIMO_SystemModel_2users_new}. \textcolor{black}{We can see from the denominator in \eqref{eq:rateMRT} that such} \ac{IUI} is determined by the similarity of \ac{ARV}s, which is analyzed in the sequel\footnote{\textcolor{black}{Since the ZF vector is the projection of the MRT precoding vector into the nullspace of the IUI vector subspace, this will be later used in the user selection and precoding design, where specific precoders based on ZF will be implemented over the candidate users.}}. In particular, we are interested in the evaluation of the product  
\begin{align}
    I(j,k)&=\frac{1}{MN}\left|\left(\left(\B{q}^H(r_j,\theta_j)\otimes\B{1}_M^T\right)\odot \textcolor{black}{\B{u}^H(\theta_j)}\right)\notag \right.\\
    & \left. \times\left(\left(\B{q}(r_k,\theta_k)\otimes\B{1}_M\right)\odot \textcolor{black}{\B{u}(\theta_k)}\right)\right|\nonumber\\
    &=\frac{1}{MN}\left|\Tr\left\lbrace\left(\B{q}^H(r_j,\theta_j)\otimes\B{1}_M^T\right)\diag(\textcolor{black}{\B{u}^H(\theta_j)})\right.\notag \right. \\
    & \left.\quad\times\left.\diag(\textcolor{black}{\B{u}(\theta_k)})\left(\B{q}(r_k,\theta_k)\otimes\B{1}_M\right) \right\rbrace\right|\nonumber\\
    &\leq \frac{1}{MN} \left|\Tr\left\lbrace\left(\B{q}^H(r_j,\theta_j)\otimes\B{1}_M^T\right)\left(\B{q}(r_k,\theta_k)\otimes\B{1}_M\right)\right\rbrace \right|\nonumber \\
    & \times \left|\Tr\left\lbrace\diag(\textcolor{black}{\B{u}^H(\theta_j)})\diag(\textcolor{black}{\B{u}(\theta_k)})\right\rbrace\right|\nonumber\\
    &=\frac{1}{N}\Big|\B{q}^H(r_j,\theta_j)\B{q}(r_k,\theta_k)\Big|\Big|\sum\nolimits_{n=1}^N\B{b}^H(\theta_{j,n})\B{b}(\theta_{k,n})\Big|.
    \label{eq:IUIfirstBound}
\end{align}
The latter expression can be further bounded by taking the summation out of the absolute value, as follows
\begin{align}
    I(j,k)
        \leq \frac{1}{N}\Big|\B{q}^H(r_j,\theta_j)\B{q}(r_k,\theta_k)\Big|\sum\nolimits_{n=1}^N \Big|\B{b}^H(\theta_{j,n})\B{b}(\theta_{k,n})\Big|.
    \label{eq:generalBound}
\end{align}
Observe that when the inner products present identical sign for all $n$, \eqref{eq:generalBound} is equal to the bound in \eqref{eq:IUIfirstBound}. This situation also arises when the angles for the different modules are the same, i.e., $\theta_{n}=\theta$. In addition, a simplified version of $I(j,k)$ in \eqref{eq:generalBound} is useful when both users lie on the FF region, or under the assumption $\theta_n\approx\theta$. 
For such a case we have that
\begin{align}
   I(j,k)&=\frac{1}{MN}|\B{q}^H(r_j,\theta_j)\B{q}(r_k,\theta_k) \otimes \B{b}^H(\theta_j)\B{b}(\textcolor{black}{\theta_k})|\notag\\
    &=\frac{1}{MN}|\B{q}^H(r_j,\theta_j)\B{q}(r_k,\theta_k) \B{b}^H(\theta_j)\B{b}(\theta_k)|\nonumber\\
    &= \frac{1}{MN}\Big|\B{q}^H(r_j,\theta_j)\B{q}(r_k,\theta_k)\Big|\Big|\B{b}^H(\theta_j)\B{b}(\theta_k)\Big|.
    \label{eq:comAngBound}
 \end{align}
In moving from \eqref{eq:generalBound} to \eqref{eq:comAngBound} we split the interference into two parts: \textit{i)} one corresponding to the module separation, and \textit{ii)} a second factor corresponding to the modules themselves. Thus, both expressions present products which depend on vectors with unit-norm elements, but very different sizes. Indeed, the scenario of interest is such that where $N\gg M$, in such a way that it makes it possible to consider the individual modules achieving the conditions of far-field propagation. In contrast, the distance between the modules incorporates the effects of the near-field propagation environment. Accordingly, the interference mitigation capabilities within each module are limited compared to that of the modular array. This is because the vector space of dimension $N$ allows much more flexibility to handle interference than the one of dimension $M$ corresponding to the modules. Therefore, the following subsections allow us to draw relevant conclusions for the approximation in \eqref{eq:generalBound} and the exact reformulation of \eqref{eq:comAngBound}.

\subsection{Inter-module interference}
We now study the first factor in \eqref{eq:generalBound} and \eqref{eq:comAngBound}, $|\B{q}^H(r_j,\theta_j)\B{q}(r_k,\theta_k)|$ to evaluate the impact of user location on the \ac{IUI}, i.e., 
\begin{equation}
     \frac{1}{N}|\B{q}^H(r_j,\theta_j)\B{q}(r_k,\theta_k)|=\frac{1}{N} \left|\sum_{n\in\mathcal{N}}e^{j\frac{2\pi}{\lambda} (r_{k, n}-r_{j,n})} \right|,
     \label{eq:interModuleInterf}
\end{equation}
where $r_{k, n}$ and $r_{j, n}$ represent the distance from the $n$-th module to user $k$ and $j$, respectively. A closed-form expression for this equation is difficult to achieve due to the complicated definition of the module distances in \eqref{eq:distanceNmodule}. Nevertheless, different approximations can be employed to arrive at analytical forms for each possible scenario, as formally stated in the following proposition.

\begin{prop}\label{propo1}
The effect of module separation on IUI can be approximated as
\begin{align}
\frac{1}{N}|\B{q}^H(r_j,\theta_j)\B{q}(r_k,\theta_k)|
\approx\frac{1}{N\sqrt{2a}}\left|F(t^{+})-F(t^{-})\right|
\label{eq:integralApprox},
\end{align}
with parameters $a>0$ and $b$ dependent on user locations, $t^{-}=-\frac{a}{\sqrt{2}}+\frac{b}{\sqrt{2a}}$ and $t^{+}=\frac{a}{\sqrt{2}}+\frac{b}{\sqrt{2a}}$, and $F(\cdot)$ being the Fresnel function $F(\tau)=\int_0^\tau[\cos\big(\frac{\pi}{2}t^2\big)+j\sin\big(\frac{\pi}{2}t^2\big)]dt$ \cite{AbSt65}. 
\end{prop}
\begin{proof}
	See Appendix~\ref{ap:Fresnel}.
\end{proof}
Some features of the Fresnel function are helpful to identify convenient ranges of values for $a$ and $b$. Specifically: \textit{i)} $F(0)=0$; \textit{ii)} $F(x)$ is an odd function which quickly increases with $x>0$; and \textit{iii)} $F(x)$ is upper bounded by $0.95$ and converges to $\frac{1}{\sqrt{2}}$ as $x$ becomes large. In the ensuing subsections, we explore the values of $a$ and $b$ for several practical situations to achieve insight regarding the influence of user locations over inter-user interference. This information will lead to later define geometry-based strategies for selecting users in a scheduling procedure.

\subsubsection{Interference in the NF region}
\label{sec:InterfNF}
When both users lie on the NF region, the second-order Taylor approach of the module distances leads to tractable expressions for the interference \cite{Sh62,CuDa22}. In particular, the parameters $a$ and $b$ of the resulting quadratic form are closely related to the geometrical features of the scenario, as follows
\begin{align}
    a&=\frac{1}{\lambda}N^2S^2d^2\left(\frac{\cos^2(\theta_k)}{r_k}-\frac{\cos^2(\theta_j)}{r_j}\right),\label{eq:paramInterfNFa}\\
    b&=-\frac{2}{\lambda}NS d(\sin(\theta_k)-\sin(\theta_j)).
    \label{eq:paramInterfNFb}
\end{align}
\textcolor{black}{As we separated the interference contribution into two terms in \eqref{eq:generalBound}, we first note that the parameters $a$ and $b$ are independent of the module size $M$. Observe from $a$ that the quadratic scaling factor depends on the physical size of the antenna array, including the number of modules, the separation among them, and the elements distance. This is then multiplied by the two quotients that depend on the user positions. In particular, by setting the quotient to a factor $\frac{\cos^2(\theta)}{r}=\phi$ we define an ellipse with its major axis starting at the BS and aligned with the x-axis \cite{CuDa22}. Thus, users placed along these two ellipses will obtain the same value for the parameter $a$. For the parameter $b$, we obtain a scaled version of the spatial angular distance $\sin(\theta_k)-\sin(\theta_j)$. As such, users located along a straight line originating from the BS at given angles will obtain the same values of $b$, regardless of their distances from the BS.}

When performing scheduling tasks, it is usually desirable to find sets of users with negligible or even zero \ac{IUI}. This can be achieved when \eqref{eq:integralApprox} is small. A reduced value for the numerator of \eqref{eq:integralApprox} can be achieved when the parameters $a$ and $b$ satisfy, simultaneously, the next conditions
\begin{align}
    b\approx a\sqrt{a},\qquad b\approx -a\sqrt{a},
    \label{eq:integralZeroCond}
\end{align}
so that $F(t^{+})-F(t^{-})\approx 0$. This situation is only met with equality when $a=b=0$. However, this situation would cause an indetermination in \eqref{eq:integralApprox} as $a$ also appears in the denominator. Also, recall that we assume $a>0$ in Appendix \ref{ap:Fresnel}, so we cannot achieve \eqref{eq:integralZeroCond} with an arbitrarily small precision. From a geometrical point of view, $a=b=0$ in \eqref{eq:paramInterfNFa} and \eqref{eq:paramInterfNFb} would imply the trivial case on which the positions of the two users are identical, that lacks practical interest. Alternatively, as the numerator of \eqref{eq:integralApprox} is bounded, the interference can be reduced by increasing the denominator
\begin{equation}
    \sqrt{2a}=\sqrt{\frac{2}{\lambda}}NSd\left(\frac{\cos^2(\theta_k)}{r_k}-\frac{\cos^2(\theta_j)}{r_j}\right)^{1/2}.
    \label{eq:NFinterference}
\end{equation}
Note that this happens when the user locations differ. For instance, selecting users such as $r_k\gg r_j$ might reduce the interference, but the channel for the furthest user will be poor due to path loss. Hence, it is interesting to select users located at distances in the same order of magnitude, $r_k\approx r_j$, but $\theta_k$ and $\theta_j$ such that the difference of cosines is as large as possible. Interestingly, the module separation $S$ increases the number of grating lobes but linearly increases the denominator of the interference function of \eqref{eq:integralApprox}, that is \eqref{eq:NFinterference}, then resulting in a trade-off. Also, a larger number of modules, i.e. increasing $N$, contributes to reducing the interference. From a geometrical perspective, we arrive at a similar conclusion, since for users located at similar distances to the BS, $r_k\approx r_j$, only the negative term in \eqref{eq:distanceNmodule} for $r_{k,n}$ and $r_{j,n}$ significantly changes from the module distances of the two interfering users.

To better illustrate our observation and provide further insight, we consider the scenario where a user is located at a distance $r_k$ and $\theta_k=0$ (i.e., located over the $x$-axis), with $r_k$ and $r_j$ being comparable. Thus, the denominator of the interference function of \eqref{eq:integralApprox}, $\sqrt{2a}$ in \eqref{eq:NFinterference}, can be approximated by
\begin{equation}
    \sqrt{2a}\approx\sqrt{\frac{2}{\lambda r_j}}NSd|\sin(\theta_j)|.
    \label{eq:sqrt2aapprox}
\end{equation}
As we desire to get large values for \eqref{eq:sqrt2aapprox} in order to reduce the interference, and under the aforementioned assumptions, we seek users with angular coordinate $\theta_j$ away from zero but, at the same time, a value $r_j$ as small as possible. 


\subsubsection{Interference in \ac{NF} and \ac{FF} regions}
\label{sec:InterfNFFF}
When users close to the \ac{BS} are not eligible to be scheduled for transmission, it is interesting to check the interference caused to those users in the region where the \ac{FF} condition holds. In this case, one of the \acp{ARV} obeys \eqref{eq:PW_delays}, and the difference of distances can be written as $r_{k,n}-r_{j,n}=r_{k,n}+nSd\sin(\theta_j)$. Using the second-order Taylor approach for the NF user as in Appendix \ref{ap:Fresnel}, we characterize the interference with the approximation  in \eqref{eq:integralApprox} by updating the values of $a$ and $b$ accordingly, that is,
\begin{align}
    a&=\frac{1}{\lambda}N^2S^2d^2\left(\frac{\cos^2(\theta_k)}{r_k}\right)\label{eq:paramInterfFFa},
\end{align}
whereas $b$ is that in \eqref{eq:paramInterfNFb}. 
If we consider the user in the NF region fixed, the condition in \eqref{eq:integralZeroCond} can be recasted to
\begin{equation*}
    \theta_j\approx\arcsin\left(\sin(\theta_k)\pm\frac{N^2S^2d^2}{2r\sqrt{r_k\lambda}}\cos^3({\theta_k})\right).
\end{equation*}
Since the assumption of a given user in the NF region keeps the denominator in \eqref{eq:integralApprox} fixed, feasible values of $\theta_j$ critically depend on $\theta_k$, as very small values on $\cos^3(\theta_k)$ are interesting to achieve a smaller amount of interference on users located in the \ac{FF} region. Due to the dependence on $r_k$, this effect is stronger when the user in the NF region is closer to the \ac{BS}.

If we consider the opposite scenario, that is, a given user in the FF and a low interference candidate in the NF counterpart, strong candidates to be scheduled are those users with a reduced value of $r_k$ and a large value for $\cos^2(\theta_k)$. In other words, users close to the broadside of the BS.  
\subsubsection{Interference in the FF region}
\label{sec:InterfFF}
When both users are located at a significant distance to the \ac{BS}, the vectors in \eqref{eq:PW_delays} accurately represent inter-module delays, and the product can be written using the well-known expression
\begin{equation}
     \frac{1}{N}|\B{q}^H(r_j,\theta_j)\B{q}(r_k,\theta_k)|=\left| \frac{\sin{\left(\tfrac{\pi}{2}(N S(\sin(\theta_k)) - \sin(\theta_j))\right)}}{N \sin{\left(\tfrac{\pi}{2}  S(\sin(\theta_k) - \sin(\theta_j))\right)}} \right|.
     \label{eq:FF_IUI}
\end{equation}
In contrast with the NF scenario in Sec. \ref{sec:InterfNF}, this interference only depends on the user angles $\theta_k$ and $\theta_j$. While in the NF scenario the module separation $S$ increases the denominator in the interference expression and hence increases \eqref{eq:NFinterference}, in the FF case $S$ denotes the number of grating lobes in the radiation pattern of the antenna array, thus reducing the angular directions experiencing low interference. Accordingly, we can expect stronger interference in the FF as the \acp{ARV} are independent of the distance, and there is no counterpart for balancing the effects of the grating lobes.

\subsection{Intra-module interference} 
Thus far, we have studied the first factor within the interference bound in \eqref{eq:generalBound}, associated with inter-module interference. We now study the second factor caused by the intra-module interaction. The general expression for this factor reads as  
\begin{align}
     &\frac{1}{M}|\B{b}^H(\theta_{j, n})\B{b}(\theta_{k,n})|\nonumber \\& = \frac{1}{M}\left|\sum_{m\in\mathcal{M}}e^{j \tfrac{2\pi}{\lambda} m d (\sin(\theta_{k, n}) - \sin(\theta_{j, n}))} \right| \nonumber \\
     & = \frac{1}{M}\left| \frac{\sin{\left(M\tfrac{\pi}{\lambda}d(\sin(\theta_{k,n}) - \sin(\theta_{j,n}))\right)}}{ \sin{\left(\tfrac{\pi}{\lambda}d  (\sin(\theta_{k,n}) - \sin(\theta_{j,n}))\right)}} \right|.
    \label{eq:interModuleInterf_2ndFactor_NFNF}
\end{align}
It is interesting to determine the geometrical region where the bound of \eqref{eq:generalBound} becomes tight. A similar region was defined in \cite{LiDoZhYoJiZh24}, but led to a very pessimistic bound that can be further enhanced. In particular, we characterize the region for which it becomes reasonable to consider that $\theta_{k,n}\approx\theta_k$, where $\theta_k$ is the angular position of the user and $\theta_{k,n}$  the angle of the user to the center of each $n$-th module. To determine such a region, we use a criterion similar to that employed to obtain the half-power beamwidth, now with a $5\%$ power reduction for increased accuracy. This is formally stated in the following Proposition.

\begin{prop}{(Common angle approximation)}\label{propo2}
\\The distance $r$ for which $\frac{1}{M}|\B{b}^H(\theta_k)\B{b}(\theta_{k,n})|\approx 0.95$ is given by
\begin{equation}
    r\geq\frac{(N-1)Sd}{2\epsilon},
    \label{eq:rApprox}
\end{equation}
with $\epsilon=-\frac{2}{M}0.18$.
\end{prop}
\begin{proof}
	See Appendix~\ref{ap:Dirichletkernel}.
\end{proof}
Note the dependency in Proposition \ref{propo2} on the number of elements per module $M$, which linearly increases the minimum distance, as the angular resolution also scales linearly with this parameter. The accuracy of this approximation is illustrated numerically in Fig. \ref{fig:Approx}, where the conservative approximation for the module $n=\frac{(N-1)}{2}$ is shown.

\begin{figure}[t]
	\centering
	\includegraphics[width=.9\columnwidth]{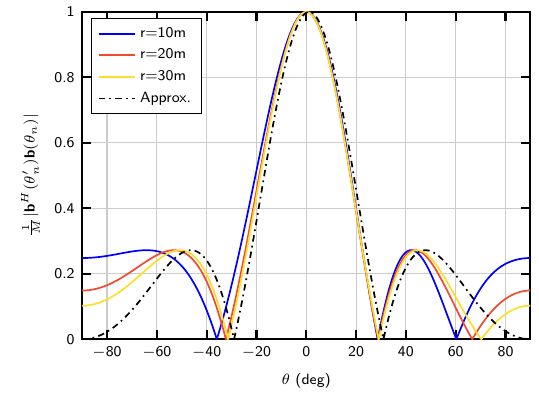}
	\caption{Worst case scenario for the common angle approximation. Parameter values are $N=32$, $M=4$,  $S=13$, $f_c=15$GHz. The proposed bound is $22.3734$ m for the considered setup. } 
 \label{fig:Approx}
\end{figure}

According to the previous discussion, the intra-module factor has to be characterized with respect to the bound \eqref{eq:rApprox}. In particular, when the distances for both users satisfy \eqref{eq:rApprox} we have that equation \eqref{eq:IUIfirstBound} (and  \eqref{eq:generalBound} as a consequence) reduces to \eqref{eq:comAngBound}, where the two factors contributing to the interference are decoupled and can be analyzed independently.

\section{User selection and precoding design}
\label{Sec:US}
The comprehensive characterization of interference patterns in gMIMO deployments based on \acp{MA} provides a solid foundation for the design of joint user selection and precoding schemes. Based on these features, two user selection algorithms are designed and described next, together with an algorithmic description for each of them.

\subsection{Algorithm 1: RSS.}
The first algorithm, referred to as \ac{RSS}, leverages the insights presented in Sec. \ref{sec:interferenceModular}, as user location in the area covered by the \ac{BS} is a fundamental factor in the eligibility for such a user. Accordingly, the \ac{RSS} method starts by selecting a rectangular-shaped area to discriminate the potential users, where $l$ sets the limit on both the horizontal and vertical axes, making an overall search area of $2l^2$. Thanks to this separation, it is possible to avoid an extensive search containing all users. As we will later see, such a restriction notably reduces the complexity with a minor performance degradation. Once the search area is determined, the user with a minimum horizontal coordinate is selected. This is in coherence with the analysis of the interference in the NF region in Sec. \ref{sec:interferenceModular}, \textcolor{black}{which} promotes users with lower horizontal coordinates with the aim of reducing \ac{IUI}. To check whether a given user stands out as a low \ac{IUI} candidate, we use the following metric
\begin{equation}
    \gamma_{\text{eq},k}=\|\B{h}_{\rm k}\|_2^2/\left(\B{h}_{\rm k}^H\big(\mathbf{I}-\sum\nolimits_{i=1}^n{\B{f}}^{(i)}{\B{f}}^{H,(i)}\big)\mathbf{h}_{\rm k}\right),
\end{equation}
where ${\B{f}}^{(i)}$ denotes a tentative precoding vector employed during the determination of the set of active users $\mathcal{S}$. In particular, at the iteration $n$ of the algorithm, the ZF precoder is computed for the selected user $k$, as follows 
\begin{align}
	\B{f}^{(n)}=\alpha_k\big(\mathbf{I}-\B{H}_{\bar{k}}(\B{H}_{\bar{k}}^H\B{H}_{\bar{k}})^{-1}\B{H}_{\bar{k}}^H\big)\B{h}_{\rm k},
	\label{eq:zf_precoders}
\end{align}
where $\B{H}_{\bar{k}}=\{\B{h}_{\rm i}\}_{i\in\mathcal{S},i\neq k}$ and $\alpha_k$ guarantees that $\B{f}^{(n)}$ is unit-norm.

The ratio $\gamma_{\text{eq},k}$ provides information regarding the amount of \ac{IUI} suffered by user $k$, with a denominator that penalizes users whose channel vectors exhibit high correlation with those already marked as active. In particular, the user is included in the set of active users $\mathcal{S}$ if $\gamma_{\text{eq},k}$ is above a given threshold $\mu$. This value can be adapted depending on the scenario, to serve users by adjusting the tolerable amount of interference. The algorithm then continues selecting users from the rectangular area earlier determined, until there are no more eligible users. Next, the rectangular area is increased and the selection procedure starts again. The method stops when the stopping criterion is met. When the procedure for joint user selection and precoding ends, the set $\mathcal{S}$ is fully determined and the actual precoders in \eqref{eq:sigr} are found as 
\begin{align}
	\B{p}_k=\alpha_k\big(\mathbf{I}-\B{H}_{\bar{k}}(\B{H}_{\bar{k}}^H\B{H}_{\bar{k}})^{-1}\B{H}_{\bar{k}}^H\big)\B{h}_{\rm k}.
\end{align}
As a final step, the power allocation coefficients $p_k$ are computed using the waterfilling procedure.

\begin{algorithm}[t]
	\caption{Rectangular-Search Scheduler (RSS)}\label{alg:RScheduler}
	\begin{algorithmic}[1]
	\small
		\STATE  $n \gets 0$, $\mathcal{S}^{(0)}\gets\emptyset$, $\mathcal{K}^{(0)}\gets\{1,\ldots,K\}$,  $l \gets$ initialization 
		\REPEAT
        \STATE $\mathcal{R}\gets \{i\in\mathcal{K}^{(n)}|x_i <l,|y_i| <l\}$
		\REPEAT
		\STATE $k\gets \min_{i\in\mathcal{R}}x_{i}$
        \STATE $\mathcal{R}\gets\mathcal{R}\setminus\{k\}$
		\STATE $\gamma_{\text{eq},k}\gets\|\B{h}_{\rm k}\|_2^2/(\B{h}_{\rm k}^H(\mathbf{I}-\sum\nolimits_{i=1}^n{\B{f}}^{(i)}{\B{f}}^{H,(i)})\mathbf{h}_{\rm k})$
        \IF{$\gamma_{\text{eq},k}\geq \mu$}
        \STATE  $n \gets n+1$
        \STATE $\mathcal{S}^{(n)}\gets\mathcal{S}^{(n-1)}\cup\{k\}$
		\STATE $\B{f}^{(n)}\gets$ Compute ZF precoder for user $k$
		\STATE $\mathcal{K}^{(n)}\gets\mathcal{K}^{(n-1)}\setminus\{k\}$
        \ENDIF
		\UNTIL $\mathcal{R}=\emptyset$
		\STATE $l \gets l +1$
		\UNTIL{$\mathcal{K}^{(n)}=\emptyset$ or stopping criterion}
	\end{algorithmic}
\end{algorithm}

\subsection{Algorithm 2: FLS.}
The method proposed in Alg. \ref{alg:RScheduler} uses a rectangular-shaped area, which represents a trade-off between the \ac{IUI} due to the modular array beam pattern and the user distance $r_k$. According to \eqref{eq:channel}, this distance is inversely proportional to the achievable SINR of the user, and determines the attainable performance. 

In order to reduce complexity, \textcolor{black}{an alternative} version of Alg. \ref{alg:RScheduler} \textcolor{black}{based on the insights of \eqref{eq:sqrt2aapprox}} and referred to as \ac{FLS} is presented in Alg. \ref{alg:FLScheduler}, where only the amount of \ac{IUI} is considered in the search of candidate users. Hence, this simplification avoids defining a rectangular area and directly evaluates users based on their position over the horizontal axis.  Specifically, for users located far away in the vertical axis, FLS will prioritize serving such users over other candidates with larger equivalent channel gains, which are closer to the \ac{BS}. \textcolor{black}{Since the set of candidate users will differ from that considered by RSS, its performance will depend on the scenario under consideration, and in the spatial distribution of users.}

\begin{algorithm}[t]
	\caption{Front Line Scheduler (FLS)}\label{alg:FLScheduler}
	\begin{algorithmic}[1]
	\small
		\STATE  $n \gets 0$, $\mathcal{S}^{(0)}\gets\emptyset$, $\mathcal{K}^{(0)}\gets\{1,\ldots,K\}$  
		\REPEAT
       	\STATE $k\gets \min_{i\in\mathcal{K}^{(n)}}x_{i}$
       	\STATE $\gamma_{\text{eq},k}\gets\|\B{h}_{\rm k}\|_2^2/(\B{h}_{\rm k}^H(\mathbf{I}-\sum\nolimits_{i=1}^n{\B{f}}^{(i)}{\B{f}}^{H,(i)})\mathbf{h}_{\rm k})$
        \IF{$\gamma_{\text{eq},k}\geq \mu$}
        \STATE  $n \gets n+1$
        \STATE $\mathcal{S}^{(n)}\gets\mathcal{S}^{(n-1)}\cup\{k\}$
		\STATE $\B{f}^{(n)}\gets$ Compute ZF precoder for user $k$
		\STATE $\mathcal{K}^{(n)}\gets\mathcal{K}^{(n-1)}\setminus\{k\}$
        \ENDIF
		\UNTIL{$\mathcal{K}^{(n)}=\emptyset$ or stopping criterion}
	\end{algorithmic}
\end{algorithm}

\subsection{A note on complexity.}
The computational cost of \eqref{eq:zf_precoders} is high, especially due to the computation of the matrix inversion and the matrix products within the inverse operation, with complexity costs in the order of $\mathcal{O}(|\mathcal{S}|^3)$ and $\mathcal{O}(|\mathcal{S}|^2MN)$, respectively, where $|\mathcal{S}|$ is the number of served users \cite{KaMuBjDe14}. Hence, these calculations become the performance bottleneck when a large number of users is checked in the selection procedure. To reduce such complexity, it is possible to approximate the computation of the inverse as in \cite{KaMuBjDe14,BeCaShTu19}, which entails a performance loss inherent to the approximation. In this work, we propose an alternative \textit{exact} procedure that takes advantage of the iterative structure of the Gram matrix $\B{H}_{\bar{k}}^H\B{H}_{\bar{k}}$. This novel procedure, which is detailed in Appendix \ref{ap:Schur}, allows to obtain the exact matrix inversion with a computational cost of about $\mathcal{O}(|\mathcal{S}|MN)$, which is linear on both the number of users and antennas.

\section{Numerical Results}
\label{Sec:NR}
In this section, we conduct numerical experiments to assess the performance benefits of the proposed methods. We consider a carrier frequency $f_c=15$ GHz in the FR3 band, due to its promising balance between bandwidth availability and coverage \cite{Emil2024}. Unless otherwise stated, we consider a scenario where the BS is equipped with a modular array with $N=32$ modules of $M=4$ antennas each with a separation parameter $S=13$, attempting to serve $K=300$ users, and the reference channel power gain is $\beta_0=0$ dB. 

According to the spatial features of the test environment defined in \cite{ITU} for evaluating 5G radio technologies in an urban environment with high user density, we consider a first scenario where users are uniformly distributed inside an area for which $r_k\in[10,60]$ m and $\theta_k\in[-60^\circ,60^\circ]$, with SNR values ranging from $0$ to $25$ dB. \textcolor{black}{The stopping criterion in Alg. \ref{alg:RScheduler} and Alg. \ref{alg:FLScheduler} is a reduction in the sum-\ac{SE}.} For benchmarking purposes, we compare our results with some recent methods, namely: (\textit{i}) \ac{DBS} in \cite{GonzalezComa2021}; (\textit{ii}) the greedy approach in \cite{LiDoZhYoJiZh24}, labeled as Greedy in the figures; and (\textit{iii}) the classical \ac{SUS} scheme \cite{GuUtDi09}. \textcolor{black}{The collocated configuration with $NM$ antennas, i.e. $S=M=1$, is analyzed as a reference. Note that, under this configuration, the antenna array aperture reduces to $(NM-1)d$.} It is worth mentioning that all these benchmark methods avoid the combinatorial nature of the problem in \eqref{eq:problem} in different ways: on the one hand, \ac{SUS} in \cite{GuUtDi09} explores all the available users to find the best-fitting candidate at each step. Conversely, the greedy method in \cite{LiDoZhYoJiZh24} (specifically designed \textcolor{black}{for} modular arrays) randomly selects a user and computes the achievable sum-\ac{SE} to ensure that a performance gain is obtained. This notably speeds up the search step, although performance is limited due to the random user selection scheme. Finally, the \ac{DBS} method in \cite{GonzalezComa2021} uses the distance from the user to the \ac{BS} as a channel quality indicator to reduce the number of candidate users. However, this strategy neglects the specific patterns of modular arrays and sectored user locations.   
\begin{figure}[t]
	\centering
	\includegraphics[width=.9\columnwidth]{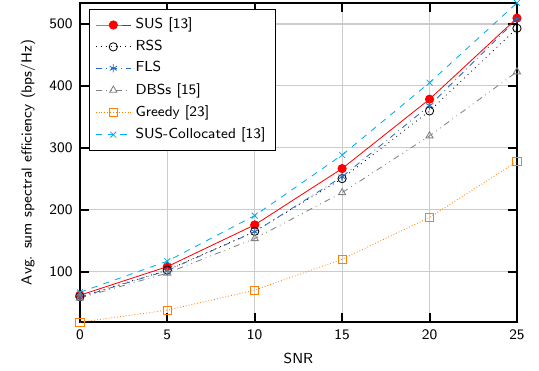}
	\caption{Sum-SE vs. SNR for the different schemes.  $N=32$, $M=4$, $K=300$, $S=13$, $f_c=15$ GHz, $r_k\in[10,60]$ m and $\theta_k\in[-60^\circ,60^\circ]$.} 
 \label{fig:achievRatesT}
\end{figure}

\subsection{Assessment of the proposed methods}
\label{sec:Perf_eval}
\textcolor{black}{We will first focus on evaluating the performance of the user selection mechanisms proposed in this work. Unless otherwise stated, we consider a MA architecture with $N=32$ and $M=4$ as in \cite{Li2023}.} Fig. \ref{fig:achievRatesT} shows the average sum spectral efficiency for the different scheduling strategies. As one can notice, the proposed methods achieve results very close to \textcolor{black}{those} of the \ac{SUS} benchmark of \cite{GuUtDi09}, which checks all the available users at each iteration. Hence, using the insights from the interference analysis performed in Sec. \ref{sec:interferenceModular}, \textcolor{black}{we define a spatial region
of interest to dramatically reduce} the number of candidate users, with a minor performance degradation. As observed in \eqref{eq:NFinterference}, users with angular coordinate $\theta$ away from $0$ suffer from less interference, and they are, therefore, easier to accommodate in the scheduling procedure. At the same time, the channel gains are directly related to the distance to the BS $r_k$, and the rectangular area employed by the \ac{RSS} proposal constitutes a compromise between these two criteria. The simpler \ac{FLS} scheme selects users disregarding the distance criterion, since the difference in the angular domain helps to alleviate the interference, thus enhancing the speed of the procedure. 
The DBS method \cite{GonzalezComa2021} only relies on the distance $r_k$ to prioritize users, thus reducing the computational cost but penalizing the overall performance results, while the \textcolor{black}{greedy} algorithm in \cite{LiDoZhYoJiZh24} leads to the worst performance result. Finally, we note that the achievable throughput with a collocated configuration provides the largest \ac{SE}. \textcolor{black}{Even though this configuration presents a poorer spatial resolution, the wide spread of the user locations in the area covered by the BS leads to milder \ac{IUI} because of the absence of grating lobes.}

\begin{figure}[t]
	\centering
	\includegraphics[width=.9\columnwidth]{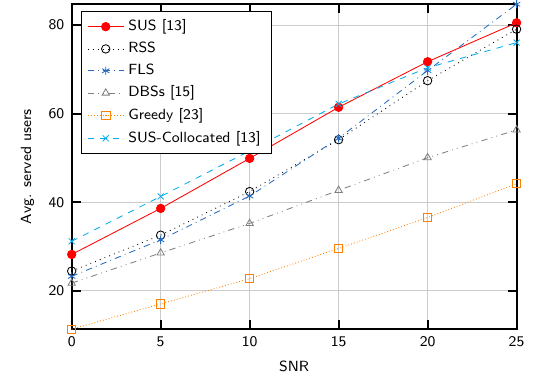}
	\caption{Number of served users vs. SNR for the different schemes.  $N=32$, $M=4$, $K=300$, $S=13$, $f_c=15$ GHz, $r_k\in[10,60]$ m and $\theta_k\in[-60^\circ,60^\circ]$.} 
 \label{fig:usersT}
\end{figure}

\begin{table}[t] 
	\centering
	\setlength{\tabcolsep}{5pt}
	\caption{Execution time (ms)}
	\begin{tabular}{|c|c|c|c|c|c|c|}
		\hline
		Method/SNR(dB)  & \textbf{0} & \textbf{5}  & \textbf{10} & \textbf{15} & \textbf{20} & \textbf{25} \\
		\hline\hline
		\textbf{SUS} & 38.2	& 50.9	&  66.3 & 78.8	& 97.3	& 111.8\\\hline
		\textbf{Greedy}	& 2.11  & 4.02 & 6.39 & 9.08 & 12.61 &20.08 \\\hline
		\textbf{RSS}	& 5.81 & 8.81 & 12.82 & 14.53 & 20.46 & 28.21 \\\hline
        \textbf{DBS}	& 4.93  & 7.07	&  9.17 & 11.45 & 14.53 & 15.40\\\hline
         \textbf{FLS} &	5.28 & 7.68 & 11.91 & 14.33 & 20.22 & 30.11\\\hline
	\end{tabular}
	\label{tab:execution_time}
\end{table}

\begin{figure}[t]
	\centering
	\includegraphics[width=.9\columnwidth]{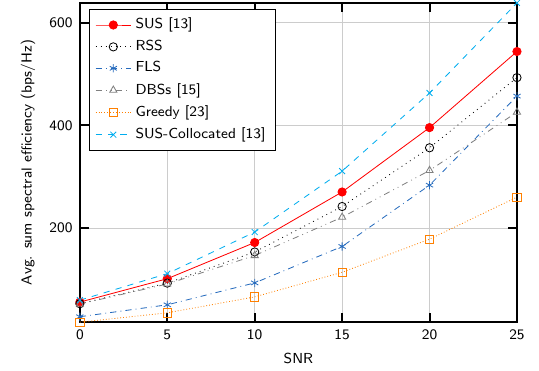}
	\caption{Sum-SE vs. SNR for the different schemes.  $N=32$, $M=4$, $K=300$, $S=13$, $f_c=15$ GHz, coordinates $x\in[10, 60]$ m, and $y\in[-60,60]$ m.} 
 \label{fig:achievRatesR}
\end{figure}

In Fig. \ref{fig:usersT}, we analyze the number of served users for each of the competing schemes. It becomes evident that curves show a similar trend as in Fig. \ref{fig:achievRatesT}, since we can expect that a larger number of served users leads to better performance. However, it is interesting to see that the use of modular arrays allows serving a larger number of users than the collocated array counterpart in the large SNR regime. This is due to the improved spatial resolution provided by modular arrays, which allows \textcolor{black}{further alleviating} the \ac{IUI}. For lower SNRs, the energy spread due to the grating lobes results in poorer channel gains for the modular arrangement, and the number of users being served is smaller than for the collocated counterpart. 

To better illustrate the benefits of the proposed user selection schemes, the execution times for all competing methods are included in Table \ref{tab:execution_time}. In all instances, the \ac{SUS} method requires the largest computational effort. The \ac{RSS} \textcolor{black}{achieves} the best trade-off performance between complexity and performance, providing over a $75\%$ reduction compared to the \ac{SUS} scheme. In general terms, the \ac{FLS} algorithm should allow for an even more reduced execution time; however, since this strategy allows to serve a larger number of users for high SNR, such a reduction may be overshadowed by the fact that complexity is proportional to the cardinality of $|\mathcal{S}|$. The reference greedy scheme developed in \cite{LiDoZhYoJiZh24} allows for a reduced execution time, at the expense of a rather low performance. Finally, the DBS also offers a limited complexity, although its performance is below the RSS and FLS schemes.

\textcolor{black}{When we compare the \ac{RSS} and \ac{FLS} strategies, we observe that the computation times are closely related to the number of served users shown in Fig. \ref{fig:usersT}, with \ac{FLS} offering better times for low SNR values, but performing slightly worse in the high SNR regime. As the number of served users for each strategy critically depends on the user distribution over the area covered by the BS, the selection of one of these strategies heavily depends on the scenario of application.} 

\begin{figure*}[t]
	\centering
	\includegraphics[width=1.8\columnwidth]{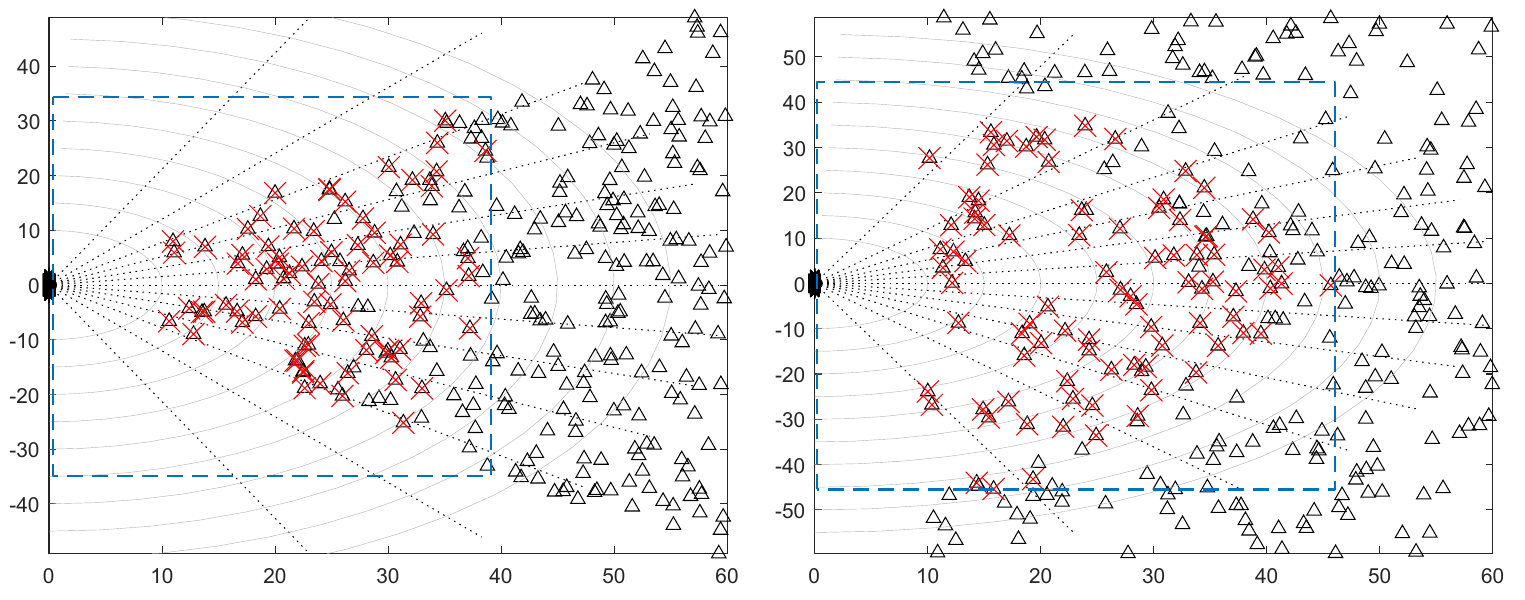}
	\caption{User distributions for the considered setups. On the left figure we show the environment defined in \cite{ITU} for evaluating 5G radio technologies in an urban environment with high user density. The right figure deploys the users in a rectangular area according to setups defined in \cite{Martín-Vega2023}.} 
 \label{fig:selection}
\end{figure*}

In Fig. \ref{fig:achievRatesR}, we now analyze the achievable performance in a second scenario inspired in \cite{Martín-Vega2023}. We consider that users are uniformly located within a rectangular area with cartesian coordinates $x\in[10, 60]$m, and $y\in[-60,60]$m, while keeping \textcolor{black}{fixed} the remaining setup parameters. Compared to the previous scenario under consideration, a larger area is covered by the \ac{BS}, which severely affects the spatial positions for the candidate users. We can observe from Fig. \ref{fig:achievRatesR} that the performance of \ac{FLS} drops, since neglects distance to the BS to select candidate users. Remarkably, \ac{RSS} still performs very close to the more complex \ac{SUS}. It is also interesting to see that having a larger area close to the \ac{BS} lessens the amount of interference. This favors the collocated configuration (which has a wider angular resolution) over the modular arrangement, exhibiting the former a superior performance.

In Fig. \ref{fig:selection}, the spatial distribution for the set of selected users is represented, for each of the scenarios under consideration. Specifically, users selected for transmission under the reference \ac{SUS} scheme are represented with red crosses, while users within the service area are depicted using black triangles. This figure is useful to confirm the intuition behind the \ac{RSS} scheme, since virtually all users selected by the \ac{SUS} scheme lie within a rectangular-shaped area, represented with dashed lines. Hence, reducing the number of candidate users has a limited impact on performance, while notably reducing the search space and hence the overall complexity.

We now aim to explore the performance of the user selection schemes as users are farther away from the BS, to assess the performance of modular arrays when the benefits of beamfocusing are less evident. First, we study the evolution of the interference pattern for two users located at the same distance $r$ to the BS. One of the users is located at $\theta=0$ whereas the other one ranges from angular locations between \textcolor{black}{$0^\circ$ and $60^\circ$}. In Fig. \ref{fig:interf}, we represent the evolution of the radiation pattern as a function of the angular location of the second user, with the case of \ac{PW} propagation included as a reference. As $r$ \textcolor{black}{increases}, the angular resolution is reduced and the behavior tends to coincide with that of the \ac{PW} case, especially for low values of $\theta$.

Based on these observations, we assess the achievable performance of the different scheduler competitors in Fig. \ref{fig:achievRatesTPW}. 
\begin{figure}[t]
	\centering
	\includegraphics[width=.9\columnwidth]{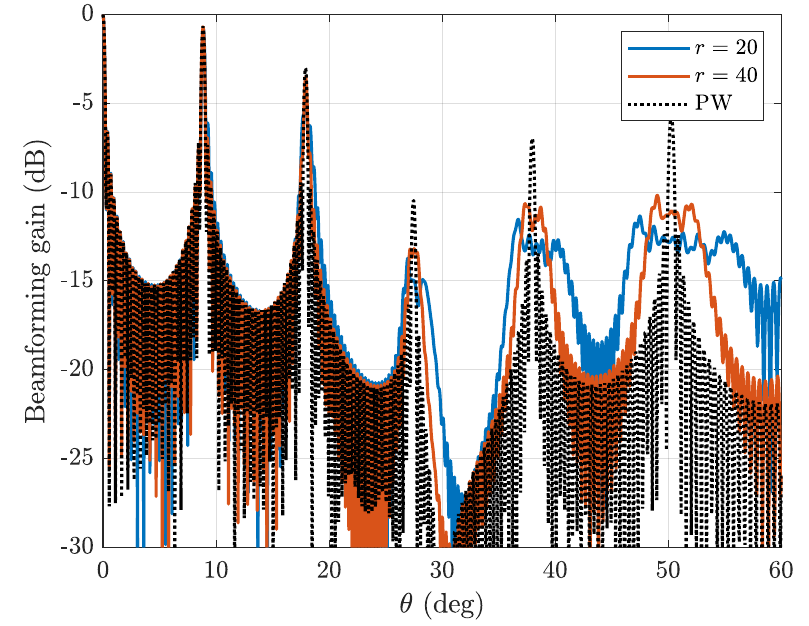}
	\caption{Interference (BF gain) suffered by a user located at a distance $r$ and angle $\theta$ when transmitting to a user located at the same distance and angle $0$, for $N=32$, $M=4$, $S=13$, $f_c=15$ GHz.} 
 \label{fig:interf}
\end{figure}
Simulation parameters are the same as in Fig. \ref{fig:achievRatesT}, but
we modify the user distance range to $r_k\in[60,150]$ m. Recall that the scheduling strategies \ac{RSS} and \ac{FLS} are designed for mitigating the \ac{IUI} in the NF region, according to the observations in Sec. \ref{sec:interferenceModular}. Hence, they cannot leverage the \ac{SW} features as users are closer to the FF region, and their achievable performance is affected. However, the \ac{RSS} technique still offers a reasonably good performance compared to all competitor schemes, while \ac{FLS} and \ac{DBS} improve the greedy scheme in \cite{LiDoZhYoJiZh24}. Noteworthy, the modular and collocated configurations exhibit a similar behavior in this scenario, since the selection of users in regions affected by the grating lobes becomes less likely.
\begin{figure}[t]
	\centering
	\includegraphics[width=.9\columnwidth]{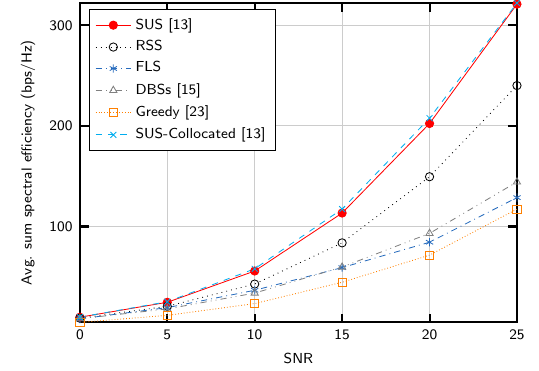}
	\caption{Sum-SE vs. SNR for the different schemes.  $N=32$, $M=4$, $K=300$, $S=13$, $f_c=15$ GHz, $r_k\in[60,120]$ m and $\theta_k\in[-60^\circ,60^\circ]$.} 
 \label{fig:achievRatesTPW}
\end{figure}

\subsection {Evaluation of different modular configurations}

\begin{figure}[htbp]
        \centering
        \includegraphics[width=.9\columnwidth]{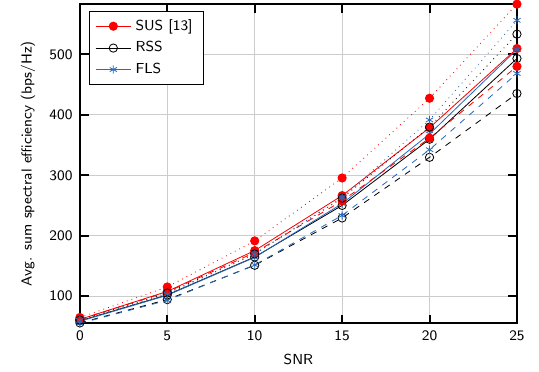}
        \caption{Sum-SE vs. SNR for different \ac{MA} schemes. $N=32$, $M=4$ (solid line) $N=4$, $M=32$ (dashed line), and $N=128$, $M=1$ (dotted line).}
        \label{fig:Comparativa_Trapz60_achievRate_withSparse}
\end{figure}

\textcolor{black}{In this section, we provide some insight regarding the influence of the antenna configuration over the performance results. To that end, we start by comparing the configuration $N=32$, $M=4$, $K=300$, $S=13$, of the previous section, with two different setups assuming an urban environment with high user density \cite{ITU2023}. In particular, we evaluate the performance of the proposed methods in a scenario with a sparse array deploying the same number of antennas $MN=128$, and of equal total size. This can be achieved by setting $N=128$, $M=1$, and $S\approx 3.2$. Further, we also include the modular array configuration $N=4$, $M=32$, and $S=125$, again resulting in the same aperture.}

\textcolor{black}{The performance results obtained for the aforementioned configurations can be seen in Fig. \ref{fig:Comparativa_Trapz60_achievRate_withSparse}, where the curve corresponding to the configuration $N=32$, $M=4$, $S=13$ is represented by a solid line, the modular configuration given by $N=4$, $M=32$, and $S=125$ is shown with a dashed line, and the sparse array $N=128$, $M=1$, $S\approx 3.2$ with a dotted line. According to \cite{LiDoZhYoJiZh24}, the spatial angular resolution of a modular array is $\frac{2}{NS}$ for $d=\lambda/2$, and the undesired grating lobes arise in the angular domain with a period of $\frac{2}{S\lambda}$. As such, the angular resolution is approximately the same for both the scenario of Sec. \ref{sec:Perf_eval} and the uniform sparse array, whereas the number of grating lobes is smaller for the sparse array. This results in better performance for the sparse antenna array. Regarding the scenario of $N=4$, $M=32$, $S=125$, the increased separation produces a large number of smaller grating lobes. This reduces the effectiveness of this antenna array configuration, leading to the worst throughput results. Observe also that the \ac{RSS} approach, intended for scenarios where $N>M$, cannot effectively handle the \ac{IUI} in the scenario with a small number of large modules, $N=4$ and $M=32$.}

\textcolor{black}{Following up with this experiment, we now evaluate different configurations for a total aperture of $4.0572$m (except for the collocated configuration, where $N=1$, $M=128$, for a reduced aperture of about $1.27$m),a fixed number of antennas $NM=128$, and a SNR of $25$dB. We consider the same user distribution as in Fig. \ref{fig:achievRatesT}, and the following \ac{MA} configurations $M=\{1,2,4,8,16,32,64,128\}$, $N=\{128,64,32,16,8,4,2,1\}$, $S=\{3.19,6.42,13,26.6,55.86,125,343,1\}$. We show the results of this experiment in Fig. \ref{fig:Comparativa_M} for the reference case of SUS strategy, since it is agnostic to the distribution of the users. For the considered scenario, configurations with large modules $M$ are generally outperformed by those with a large number of modules $N$. Analogously, reduced module separations $S$ provide better performance than configurations with farther separated modules. In other words, the increased spatial resolution provided by separating the antenna elements is mitigated by the abundance of significant grating lobes. }

\begin{figure}[t]
	\centering
	\includegraphics[width=.9\columnwidth]{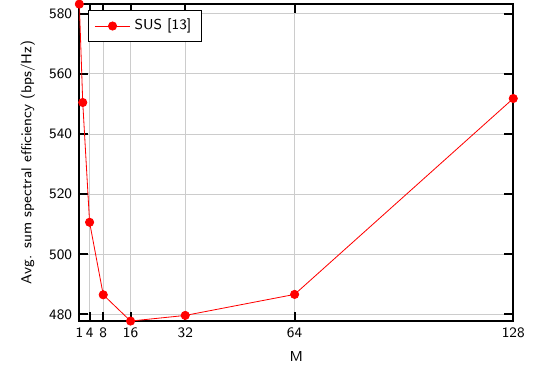}
        \caption{Sum-SE vs. $M$ for different configurations, with SNR$=25$dB.}
        \label{fig:Comparativa_M}
\end{figure}

\section{Concluding remarks}
\label{Sec:Conclusions}
In this work, we characterized the \ac{IUI} generated in a communication system composed of a modular {gMIMO} array, being able to quantify the impacts of inter-module and intra-module interferences inherent to this novel configuration. This result was then used to address the problem of user selection and precoding, designing feasible low-complexity strategies that leverage the spatial interference patterns when modular arrays are used. Performance was evaluated in a number of scenarios of interest in the FR3 band, showing that the proposed algorithms behave close to the reference SUS scheme, and largely outperform existing methods in the literature. Results confirm that the narrower spatial resolution of modular arrays allows to serve a larger number of users, and that the techniques to reduce the cardinality of the search space of potentially eligible users provide a significant decrease in the computational burden. \textcolor{black}{With regard to the benefits offered by MAs, the ultimate choice for the modular arrangement would depend on the physical restrictions for the module deployment; however, results suggest that the number of modules should be minimized whenever possible for improved SE performance.}

\textcolor{black}{MAs provide improved flexibility to implement antenna arrays when physical deployment constraints prevent from using collocated or sparse solutions. While the distributed nature of MAs brings new challenges related to inter-module synchronization and calibration, they open the door to implementing locally-distributed MIMO arrays (e.g. distribute modules sparsely in a rooftop or a facade) with reduced operational complexity compared to cell-free MIMO counterparts deployed in wider areas. While the consideration of LOS propagation is a good approximation in outdoor deployments in the FR3 band \cite{Miao2025}, the design of user selection mechanisms in richer scattering conditions is likely to bring additional challenges that deserve further attention.
}

\appendices
\section{Fresnel Integral}
\label{ap:Fresnel}
First, we apply the second-order Taylor approach to $r_{k,n}$ and $r_{j,n}$ as as function of the module index $n$, leading to $r_{k,n}-r_{j,n}\approx \frac{\lambda}{2}(\bar{a}n^2+\bar{b}n)$. Then, the right-hand side of \eqref{eq:interModuleInterf} is approximated as
\begin{align}
\frac{1}{N} \left|\sum_{n\in\mathcal{N}}e^{j\frac{2\pi}{\lambda} (r_{k,n}-r_{j,n})} \right|&\approx\frac{1}{N}\left|\sum_{n\in\mathcal{N}}e^{j\pi(\bar{a}n^2+\bar{b}n)} \right|\notag\\
&\approx\frac{1}{N}\left|\int_{-\frac{1}{2}}^{\frac{1}{2}}e^{j\pi (ax^2+bx)} dx\right|,\label{eq:integral}
\end{align}
Note that we can assume $\bar{a}>0$ without loss of generality. Further, in the second line of \eqref{eq:integral} the summation is approximated by a continuous integral in $x\in[-\frac{1}{2},\frac{1}{2}]$ such that $x=\frac{n}{N}$ \cite{LuZe21}, 
with $a=N^2\bar{a}$ and $b=N\bar{b}$. Note that this approximation holds tight for large values of $N$. At this point, it is useful to rewrite the exponent of the integral as follows
\begin{align}  
\left|\int_{-\frac{1}{2}}^{\frac{1}{2}}e^{j\pi (ax^2+bx)} dx\right|  
=\left|\int_{-\frac{1}{2}}^{\frac{1}{2}}e^{j\pi a\big(x+\frac{b}{2a}\big)^2} dx\right|\notag
\end{align}
which can be expressed, after a change of variables, as%
\begin{align*}
   \int_{-\frac{1}{2}}^{\frac{1}{2}}e^{j\pi a\big(x+\frac{b}{2a}\big)^2} dx \nonumber \\
   & \hspace{-1cm} = \frac{1}{\sqrt{2a}}\int_{t^{-}}^{t^{+}}\left[\cos\left(\frac{\pi}{2}t^2\right)+j\sin\left(\frac{\pi}{2}t^2\right)\right]dt
\end{align*}
\noindent with the integration limits $t^{-}$ and $t^{+}$ as defined in Prop. \ref{propo1}. Finally, using the definition of the Fresnel function $F(x)$ we get
\begin{align*}    \left|\int_{-\frac{1}{2}}^{\frac{1}{2}}e^{j\pi (ax^2+bx)} dx\right|= \frac{1}{\sqrt{2a}}\left|F(t^{+})-F(t^{-})\right|,
\end{align*}
which completes the proof.

\section{Geometric area for the approximation in Prop. \ref{propo2}}
\label{ap:Dirichletkernel}
We aim at finding the geometric area where the approximation of Prop. \ref{propo2} is tight, i.e., $\frac{1}{M}|\B{b}^H(\theta)\B{b}(\theta_{n})|\approx 0.95$. By letting $\sin(\theta_n)=\sin(\theta)+\epsilon$, and using $d=\lambda/2$ for simplicity, the inner product of these vectors can be written as
\begin{align*}
    \left| \frac{\sin{\left(\tfrac{\pi}{2}M(\sin(\theta_n)-\sin(\theta) )\right)}}{M\sin{\left(\tfrac{\pi}{2}( \sin(\theta_n)-\sin(\theta)) \right)}} \right| & =\left| \frac{\sin{\left(\tfrac{\pi}{2}M\epsilon\right)}}{M\sin{\left(\tfrac{\pi}{2} \epsilon \right)}} \right|\\
    &\approx\left| \frac{\sin{\left(\tfrac{\pi}{2}M\epsilon\right)}}{M\tfrac{\pi}{2} \epsilon} \right|,
\end{align*}
where we have considered a small offset $\epsilon\approx 0$ for the approximation \cite{Bjorson2024}. Under this approach we compute the value of $\epsilon$ which provides the desired accuracy, i.e.,  $\sin(0.18\pi)/(0.18\pi)\approx 0.95$, leading to $\epsilon=\pm\frac{2}{M}0.18$. Next, we can exploit the geometrical relationship between the user angle $\theta$ and the module angle $\theta_n$, as follows
\begin{equation*}
    \sin(\theta_n) = \frac{r\sin(\theta)-nSd}{r_n}
    =\sin(\theta)+ \frac{-nSd-\Delta r_n\sin(\theta)}
    {r+\Delta r_n}
\end{equation*}
where $\Delta r_n=r_n-r$. For the sake of notation simplicity, we assume that $\sin(\theta)>0$ and, accordingly, the larger angular differences lie on the distant modules, that is, those with $n<0$. Then, using the fact that $\Delta r_n>0$ in this scenario, we have that $-nSd\leq\frac{(N-1)}{2}Sd$, $-\Delta r_n\sin(\theta)\leq 0$ and $r+\Delta r_n\geq r$, thus establish the following bound  
\begin{align}
   \frac{-nSd-\Delta r_n\sin(\theta)}{r+\Delta r_n}&\leq \frac{(N-1)Sd}{2r}.\label{eq:sinShiftBound}
\end{align}
Equating this bound to $\epsilon$ the result in \eqref{eq:rApprox} follows.

\section{Low-complexity Gram matrix inversion}
\label{ap:Schur}
This matrix inversion procedure exploits the iterative nature of the proposed algorithms, which successively include a new user at each step. Then, at the ($n+1$)-th iteration of the algorithm, the composite channel matrix of \eqref{eq:zf_precoders} can be written as
\begin{equation*}
	\B{H}^{(n+1)}_{\bar{k}}=\begin{bmatrix}
		\B{H}^{(n)}_{\bar{k}}\\
		\B{h}_{\rm k}^H
	\end{bmatrix}
\end{equation*}
where $\B{H}^{(n)}_{\bar{k}}\in\mathbb{C}^{n\times MN}$ comprises the channels of the users selected in the $(n-1)$ former iterations, and $\B{h}_{\rm k}$ is the channel selected at the iteration $n$. To compute the precoders in \eqref{eq:zf_precoders}, it is necessary to obtain the Gram matrix $\B{H}^{(n+1)}_{\bar{k}}(\B{H}^{(n+1)}_{\bar{k}})^H$, and perform the matrix inversion 
\begin{equation}
	\B{G}^{(n+1)}=\left(\B{H}^{(n+1)}_{\bar{k}}(\B{H}^{(n+1)}_{\bar{k}})^H\right)^{-1}
\end{equation}
To alleviate the computational load, we propose to perform a block decomposition of the Gram matrix, leading to
\begin{equation*}
	\B{G}^{(n+1)}=\left(\begin{bmatrix}
		\B{H}^{(n)}_{\bar{k}}(\B{H}^{(n)}_{\bar{k}})^H & \B{H}^{(n)}_{\bar{k}}\B{h}_{\rm k}\\
		\B{h}_{\rm k}^H(\B{H}^{(n)}_{\bar{k}})^H & \|\B{h}_{\rm k}\|^2
	\end{bmatrix}\right)^{-1}.
\end{equation*}
Observe that the block $\B{H}^{(n)}_{\bar{k}}(\B{H}^{(n)}_{\bar{k}})^H$ was already included in the previous iteration of the algorithm, since $\B{G}^{(n)}=(\B{H}^{(n)_{\bar{k}}}(\B{H}^{(n)}_{\bar{k}})^H)^{-1}$. We now apply the Schur complement to rewrite $\B{G}^{(n+1)}
$ as follows \cite{bernstein2005matrix}
\begin{equation*}
    \B{G}^{(n+1)}=\upsilon\begin{bmatrix}
  \tfrac{1}{\upsilon}\B{G}^{(n)}+\B{G}^{(n)}\B{\xi}\B{\xi}^H\B{G}^{(n)} & -\B{G}^{(n)}\B{\xi}\\
		-\B{\xi}^H\B{G}^{(n)} & 1
	\end{bmatrix},
\end{equation*}
where we introduced the auxiliary vector $\B{\xi}=\B{H}^{(n)}\B{h}_{\rm k}$, and $\upsilon^{-1}=\|\B{h}_{\rm k}\|^2-\B{\xi}^H\B{G}^{(n)}\B{\xi}$, is the Schur complement of the block $\B{G}^{(n)}$ of the matrix $\B{G}^{(n+1)}$. From the former equality, it is clear that the block-wise inversion only relies on the calculation of matrix-vector products. Moreover, the vector $\B{\xi}$ and the product $\B{G}^{(n)}\B{\xi}$ can be computed only once, thus reducing the computational load. According to these observations, the computational cost of obtaining $\B{G}^{(n+1)}$ is in the order of $\mathcal{O}(|\mathcal{S}|MN)$, where the number of selected users $|\mathcal{S}|$ is the same as the iteration number $n$.

\bibliographystyle{IEEEtran}

\bibliography{References_XL_MIMO}

\end{document}